\documentclass[letterpaper,10pt,conference]{ieeeconf}
\IEEEoverridecommandlockouts
\usepackage{graphics}
\usepackage{epsfig}
\usepackage{mathptmx}
\usepackage{times}
\usepackage{amsmath}
\usepackage{amssymb}
\usepackage{microtype}
\usepackage{subfigure}
\usepackage{booktabs}
\usepackage{hyperref}
\usepackage{comment}
\usepackage{color}

\newtheorem{theorem}{Theorem}
\newtheorem{proposition}{Proposition}

\newtheorem{corollary}{Corollary}
\newtheorem{remark}{Remark}
\newtheorem{property}{Property}

\title{\LARGE\bf
A Game Theoretic Analysis of LQG Control under Adversarial Attack
}

\author{Zuxing Li, Gy\"{o}rgy D\'{a}n, and Dong Liu%
\thanks{This work was partly supported by the Swedish Foundation for Strategic Research (SSF) through the CLAS project and by MSB through the CERCES project.}%
\thanks{Z. Li, G. D\'{a}n, and D. Liu are with the School of Electrical Engineering and Computer Science, KTH Royal Institute of Technology, Stockholm, Sweden
{\tt\small\{zuxing;gyuri;doli\}@kth.se}}%
}

\begin{document}

\maketitle
\thispagestyle{empty}
\pagestyle{empty}

\begin{abstract}
Motivated by recent works addressing adversarial attacks on deep reinforcement learning, a deception attack on linear quadratic Gaussian control is studied in this paper. In the considered attack model, the adversary can manipulate the observation of the agent subject to a mutual information constraint. The adversarial problem is formulated as a novel dynamic cheap talk game to capture the strategic interaction between the adversary and the agent, the asymmetry of information availability, and the system dynamics. Necessary and sufficient conditions are provided for subgame perfect equilibria to exist in pure strategies and in behavioral strategies; and characteristics of the equilibria and the resulting control rewards are given. The results show that pure strategy equilibria are informative, while only babbling equilibria exist in behavioral strategies. Numerical results are shown to illustrate the impact of strategic adversarial interaction.
\end{abstract}

\section{INTRODUCTION}
\label{sec1}
Deep reinforcement learning (DRL) has recently emerged as a promising solution for solving large Markov decision processes (MDPs) and partially observable MDPs (POMDPs), thanks to deep neural networks used as policy approximators~\cite{mnih2015}. DRL has, however, been shown to be vulnerable to small perturbations of the state observation, called adversarial examples, which were found to mislead the control agent to take suboptimal control actions~\cite{huang2016}. While there has been a significant recent interest in the design of adversarial examples against DRL~\cite{huang2016,lin2017,behzadan2017,russo2019,gleave2019}, there has been little work on characterizing the ability of agents to adapt to those.

Recent work proposed to use adversarial examples for making DRL agents more robust to perturbations, by letting the adversary and the agent play against each other, and formulating the interaction as a stochastic game (SG)~\cite{pinto2017}. Nonetheless, in the case of adversarial examples the agent cannot observe the system state directly, neither can the adversary affect the state transition probabilities directly, only through the actions taken by the agent. Hence, the SG model does not capture the information structure of the problem. Effectively, in the presence of adversarial examples the agent has to solve a POMDP, where the observations are subverted by the adversary so as to mislead the agent.

As a model of this interaction, in this paper we propose a game theoretical model to study the strategic interaction between an agent that has to solve a linear quadratic Gaussian (LQG) control problem, and an adversary that can manipulate the agent's observations by a randomly chosen affine transformation subject to a mutual information constraint, and aims at minimizing the control reward. The resulting problem is formulated as a dynamic cheap talk game, which captures information asymmetry, the beliefs of the adversary and the agent, and the undetectability constraint imposed on the adversarial attacks.

Our paper contributes to the solution of the formulated game theoretical problem in two ways. First, we address necessary and sufficient conditions for the existence of subgame perfect equilibria (SPEs) in pure strategies and in behavioral strategies, and we characterize the equilibrium strategies. Second, we characterize the rewards achievable in equilibria, and relate them to the rewards of a naive agent and an alert agent under attack. The key novelty of our contribution is that we characterize the strategies to be followed by the agent and by the adversary under strategic interaction, which has not been addressed by the existing literature.


The rest of the paper is organized as follows. In Section~\ref{sec1b} we review related work. In Section~\ref{sec2} we present the system model and problem formulation. In Section~\ref{sec4} we provide analytical results. In Section~\ref{sec5} we provide numerical results. Section~\ref{sec6} concludes the paper.

{\bf Notation:} Unless otherwise specified, we denote a random variable by a capital letter and its realization by the corresponding lower-case letter. We denote  by $\mathcal{N}(\cdot,\cdot)$ the Gaussian distribution,  by $\mathbb{S}(\cdot)$ the support set, by $I(\cdot;\cdot)$ the mutual information, and by  $||\cdot||$ the cardinality of a set.

\section{RELATED WORK}
\label{sec1b}
Related to our work are previous researches on robust POMDP under uncertainty of the system dynamics~\cite{osogami2015}. In~\cite{osogami2015} the control action was optimized under the worst-case assumption of the system dynamics in each stage, i.e., the agent plays as the leader and the dynamic system plays as the follower in a Stackelberg game. SG and partially observable SG (POSG) were used to model the strategic interaction of players in a dynamic system, and have been employed in robust and adversarial problems~\cite{pinto2017,gleave2019,horak2017}. But unlike in the case of learning under adversarial attacks, in SG and in POSG the players interact with each other through the impact of their actions on the state transitions, not on the state observations. Our work is related to the cheap talk game~\cite{crawford1982}, where a sender with private information sends a message to a receiver and the receiver takes an action based on the received message and based on its belief on the inaccessible private information. Closest to our model are~\cite{saritas2017,saritas2019,zuxing2019}. In~\cite{saritas2017} a dynamic cheap talk game was proposed to study a deception attack on a Markovian system, where the actions do not affect the state transitions. In~\cite{saritas2019} authors developed a dynamic game model of the attacker-defender interaction, and characterized the optimal attack strategy as a function of the defense strategy, allowing for a static optimal defense strategy. In our preliminary work~\cite{zuxing2019} we proposed a dynamic cheap talk framework to model deception attacks on a general MDP, and addressed computational issues.

Adversarial variants of LQG control were considered in a number of recent works. A Stackelberg game was formulated in~\cite{sayin2017}, where the dynamic system is the leader, while the agent is the follower and may be an adversary. The authors formulated a finite horizon hierarchical signaling game between the sender and the receiver in a dynamic environment and showed that linear sender and receiver strategies can yield the equilibrium~\cite{sayin2019}. In~\cite{zhang2017}, the opposite problem was studied but without considering the complete strategic interaction, where the adversary optimally manipulates the control actions instead of the system states. The optimal attack on both the system state and the control action in LQG control was studied in~\cite{zhang2019}. In~\cite{chen2016}, a targeted attack strategy was studied to mislead the LQG system to a particular state while evading detection.

\section{ADVERSARIAL LQG CONTROL PROBLEM}
\label{sec2}

\begin{figure}[tb]
\begin{center}
\centerline{\includegraphics[scale=0.6]{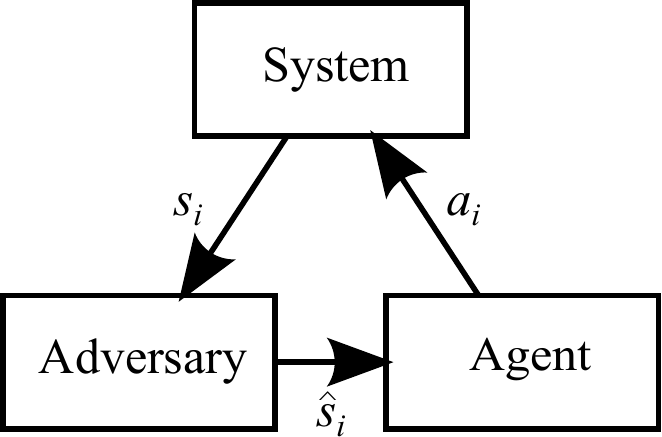}}
\caption{Considered adversarial attack on dynamic system control. In the $i$-th stage, the adversary observes the true system state $s_{i}$ and presents the manipulated state $\hat{s}_{i}$ to the agent. The agent does not have access to the true state $s_{i}$ but takes an action $a_{i}$ upon observing $\hat{s}_{i}$.}
\label{figure1}
\end{center}
\end{figure}

We consider an $N$-stage LQG control problem under adversarial attack, as illustrated in Fig.~\ref{figure1}. The system states $\{s_{i}\}_{i=1}^{N}$, the manipulated states $\{\hat{s}_{i}\}_{i=1}^{N}$, the actions $\{a_{i}\}_{i=1}^{N}$, and the instantaneous rewards $\{r_{i}\}_{i=1}^{N}$ for $1\leq i\leq N$ are described by
\begin{gather}
s_{i+1}=\alpha_{i}s_{i}+\beta_{i}a_{i}+z_{i},\textnormal{ given }\alpha_{i}\not=0,\textnormal{ }\beta_{i}\not=0;\\
\hat{s}_{i}=\pi_{i}s_{i}+c_{i};\\
a_{i}=\kappa_{i}\hat{s}_{i}+\rho_{i};\\
r_{i}=R_{i}(s_{i},a_{i})=-\theta_{i}s_{i}^{2}-\phi_{i}a_{i}^{2},\textnormal{ given }\theta_{i}>0,\textnormal{ }\phi_{i}>0;\\
S_{1}\sim b_{1}\triangleq\mathcal{N}(\mu_{1},\sigma_{1}^{2}),\textnormal{ given }\mu_{1},\textnormal{ }\sigma_{1}^{2}>0;\\
Z_{i}\sim \mathcal{N}(0,\omega_{i}^{2}),\textnormal{ given }\omega_{i}^{2}>0;\\
C_{i}\sim\mathcal{N}(0,\delta_{i}^{2}).
\end{gather}

\subsection{LQG Recapitulation}
If $\{\pi_{i}\}_{i=1}^{N}$ and $\{\delta_{i}^{2}\}_{i=1}^{N}$ are known by the agent, the above problem is a standard LQG control. In the $i$-th stage, the agent observes $\hat{s}_{i}$ but not $s_{i}$, and determines the action $a_{i}$ with the aim to maximize its expected accumulated reward. Note that it is sufficient to consider an affine function of $\hat{s}_{i}$ for $a_{i}$ in the standard LQG problem as the optimal action $a_{i}^{\star}$ is a linear function of the mean of the Gaussian posterior distribution of $S_{i}$ for the agent after observing $\{\hat{s}_{k}\}_{k=1}^{i}$ and $\{a_{k}\}_{k=1}^{i-1}$~\cite{soderstrom2002}. To compute the optimal coefficients $\kappa_{i}^{\star}$ and $\rho_{i}^{\star}$, we first define $\tilde{\theta}_{N+1}=0$, and for $1\leq i\leq N$ define $\tilde{\theta}_{i}$ as
\begin{gather}
\tilde{\theta}_{i}=\theta_{i}+\tilde{\theta}_{i+1}\alpha_{i}^{2}-\frac{\tilde{\theta}_{i+1}^{2}\alpha_{i}^{2}\beta_{i}^{2}}{\phi_{i}+\tilde{\theta}_{i+1}\beta_{i}^{2}}.
\end{gather}
Furthermore, we denote by $b_{i}\triangleq\mathcal{N}(\mu_{i},\sigma_{i}^{2})$ the belief of the agent about $S_{i}$, which is the Gaussian posterior distribution of $S_{i}$ for the agent after observing $\{\hat{s}_{k}\}_{k=1}^{i-1}$ and $\{a_{k}\}_{k=1}^{i-1}$. Then, given the manipulated state $\hat{s}_{i}$, the optimal action can be expressed as
\begin{gather}
\kappa_{i}^{\star}=-\frac{\tilde{\theta}_{i+1}\alpha_{i}\beta_{i}\pi_{i}\sigma_{i}^{2}}{(\phi_{i}+\tilde{\theta}_{i+1}\beta_{i}^{2})(\pi_{i}^{2}\sigma_{i}^{2}+\delta_{i}^{2})},\\
\rho_{i}^{\star}=-\frac{\tilde{\theta}_{i+1}\alpha_{i}\beta_{i}\mu_{i}\delta_{i}^{2}}{(\phi_{i}+\tilde{\theta}_{i+1}\beta_{i}^{2})(\pi_{i}^{2}\sigma_{i}^{2}+\delta_{i}^{2})},\\
a_{i}^{\star}=\kappa_{i}^{\star}\hat{s}_{i}+\rho_{i}^{\star}=-\frac{\tilde{\theta}_{i+1}\alpha_{i}\beta_{i}}{\phi_{i}+\tilde{\theta}_{i+1}\beta_{i}^{2}}\frac{\pi_{i}\sigma_{i}^{2}\hat{s}_{i}+\mu_{i}\delta_{i}^{2}}{\pi_{i}^{2}\sigma_{i}^{2}+\delta_{i}^{2}},
\end{gather}
where the coefficients $\kappa_{i}^{\star}$ and $\rho_{i}^{\star}$ depend on $b_{i}$; $\frac{\pi_{i}\sigma_{i}^{2}\hat{s}_{i}+\mu_{i}\delta_{i}^{2}}{\pi_{i}^{2}\sigma_{i}^{2}+\delta_{i}^{2}}$ is the mean of the Gaussian posterior distribution of $S_{i}$ for the agent after observing $\{\hat{s}_{k}\}_{k=1}^{i}$ and $\{a_{k}\}_{k=1}^{i-1}$. Note that when $\pi_{i}\equiv1$ and $\delta_{i}^{2}\equiv0$ the LQG strategy reduces to the linear quadratic regulator (LQR) strategy
\begin{align}
(\kappa_{i}^{\star},\rho_{i}^{\star})=\left(-\frac{\tilde{\theta}_{i+1}\alpha_{i}\beta_{i}}{\phi_{i}+\tilde{\theta}_{i+1}\beta_{i}^{2}},0\right).
\label{eq:LQR}
\end{align}

\subsection{Adversarial Model}
The adversary can manipulate the observation of the agent and its objective is to minimize the agent's expected accumulated reward, similar to~\cite{huang2016,russo2019}. In the $i$-th stage, the adversary chooses the manipulation parameters $\pi_{i}$, $\delta_{i}^{2}$, manipulates the state $s_{i}$ to $\hat{s}_{i}$, and then reports the manipulated state $\hat{s}_{i}$ to the agent. We consider that the adversarial manipulation is ``small", since a large manipulation may be easily detected and may also involve a high manipulation cost. Given the agent's belief $b_{i}\triangleq\mathcal{N}(\mu_{i},\sigma_{i}^{2})$, we impose the following constraints on the manipulation:
\begin{gather}
-\infty<\varepsilon'\leq\pi_{i}\leq\varepsilon<\infty;\label{eq14}\\
I(\hat{S}_{i};S_{i})=\frac{1}{2}\log\frac{\pi_{i}^{2}\sigma_{i}^{2}+\delta_{i}^{2}}{\delta_{i}^{2}}\geq\frac{1}{2}\log\lambda>0,\textnormal{ }\forall\pi_{i},\delta_{i}^{2},\label{eq15}
\end{gather}
i.e., $\frac{\pi_{i}^{2}\sigma_{i}^{2}+\delta_{i}^{2}}{\delta_{i}^{2}}\geq\lambda>1$. The mutual information constraint (\ref{eq15}) implies that the manipulated state conveys at least a certain amount of information about the system state to the agent. A larger value of $\lambda$ means a weaker adversary, and vice versa. Note that in order to satisfy the mutual information constraint, the adversary {\it cannot} use $\pi_{i}=0$. We denote by $\mathcal{A}_{i}(b_{i},\varepsilon',\varepsilon,\lambda)$ the set of feasible adversarial actions $(\pi_{i},\delta_{i}^{2})$ in the $i$-th stage subject to~(\ref{eq14})-(\ref{eq15}).
Finally, in the end of this stage, the adversary reveals $\pi_{i}$, $\delta_{i}^{2}$ to the agent\footnote{This assumption is strong but may hold in some cases. For instance, the player in the shell game reveals the cup in which the pellet is after each round.}, so as  to keep the adversarial model consistent with the standard LQG control. 

\begin{figure}[t!]
\begin{center}
\centerline{\includegraphics[scale=0.4]{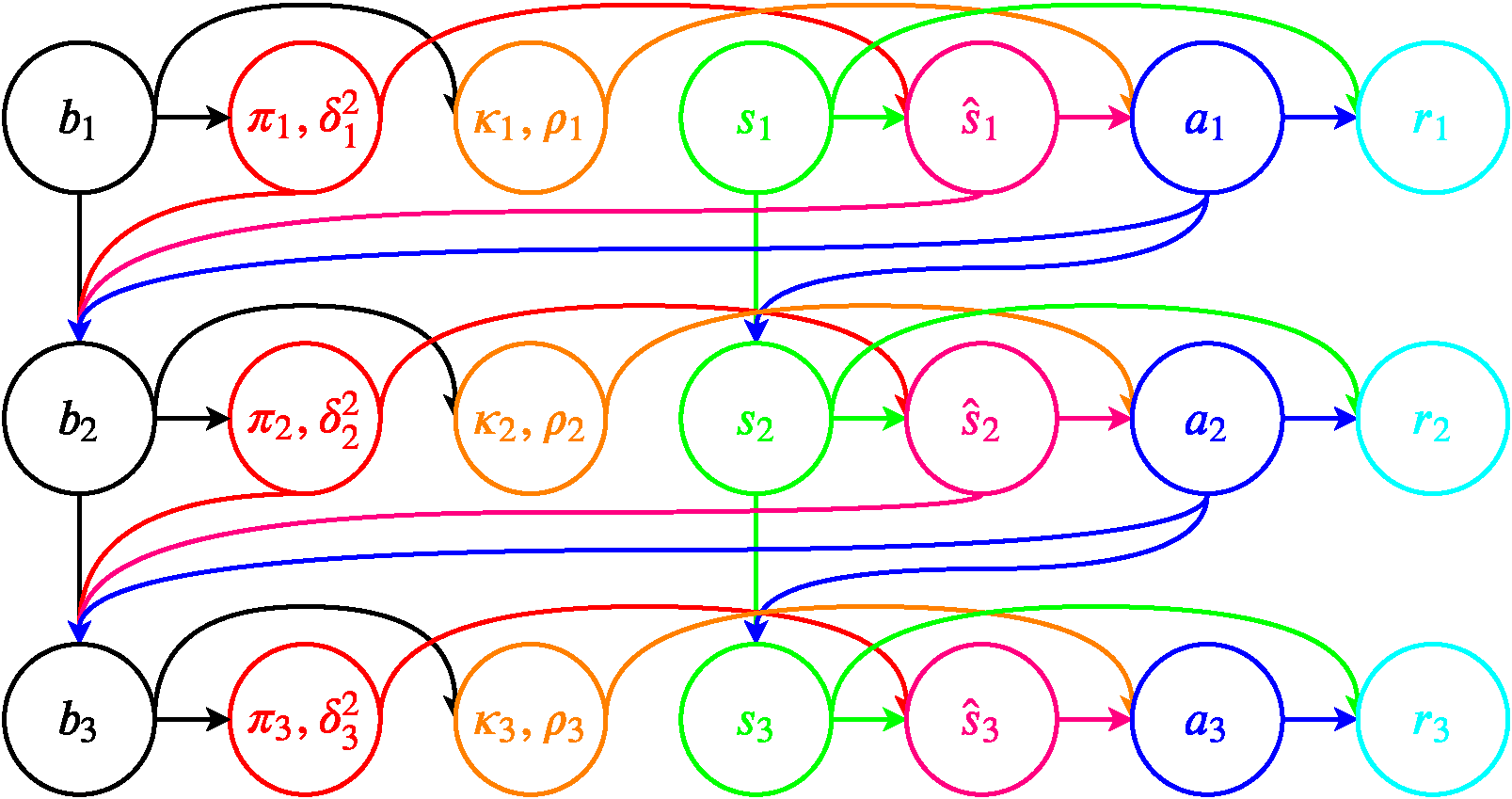}}
\caption{Illustration of information structure for a three-stage ALQG game.}
\label{figure2}
\end{center}
\end{figure}

\subsection{Adversarial LQG Control Game}
In every stage of the adversarial problem, there is a cheap talk interaction, where the adversary acts as the sender and the agent as the receiver. Different from the dynamic cheap talk game with an action-independent Markovian system~\cite{saritas2017}, we propose a novel dynamic cheap talk game to model the strategic interaction of the adversary and the agent with asymmetric information in the adversarial LQG control problem. Unlike in recent works on adversarial reinforcement learning~\cite{huang2016,lin2017,behzadan2017,russo2019}, in our model of strategic interaction the agent is aware of and can adapt to the adversary.

The game is played between the adversary and the agent over $N$ stages. In the $i$-th stage, the belief of the agent $b_{i}$ is known to the adversary and determines the action set $\mathcal{A}_{i}(b_{i},\varepsilon',\varepsilon,\lambda)$. The adversary uses a behavioral strategy $g_{i}(\pi_{i},\delta_{i}^{2}|b_{i})$ over $\mathcal{A}_{i}(b_{i},\varepsilon',\varepsilon,\lambda)$ for choosing $(\pi_{i}$, $\delta_{i}^{2})$. Then, given the observed system state $s_{i}$ it generates the manipulated state $\hat{s}_{i}$ with the probability measure $\mathcal{N}(\hat{s}_{i}|\pi_{i}s_{i},\delta_{i}^{2})$. 
The agent uses a pure strategy\footnote{The following analysis will show that it is sufficient to consider a pure agent strategy with an affine form.} $f_{i}(b_{i})$ for choosing $\kappa_{i}$ and $\rho_{i}$ based on the belief $b_{i}$, and takes the action $a_{i}=\kappa_{i}\hat{s}_{i}+\rho_{i}$ once it receives the manipulated state $\hat{s}_{i}$. Finally, the players compute the belief $b_{i+1}\triangleq\mathcal{N}(\mu_{i+1},\sigma_{i+1}^{2})$ based on the current belief $b_{i}$, the coefficient $\pi_{i}$, the variance $\delta_{i}^{2}$, the manipulated state $\hat{s}_{i}$, and the action $a_{i}$ as $\mu_{i+1}=\Lambda_{\mu}(b_{i},\pi_{i},\delta_{i}^{2},\hat{s}_{i},a_{i})$ and $\sigma_{i+1}^{2}=\Lambda_{\nu}(b_{i},\pi_{i},\delta_{i}^{2})$.

We can thus express the expected accumulated agent reward using the adversarial strategies $g^{N}\triangleq(g_{1},\dots,g_{N})$ and the agent's strategies $f^{N}\triangleq(f_{1},\dots,f_{N})$ over $N$ stages as
\begin{equation}
V\left(b_{1},g^{N},f^{N}\right)=E_{b_{1},g^{N},f^{N}}\left(\sum_{j=1}^{N}R_{j}(S_{j},A_{j})\right).
\label{eq2}
\end{equation}
Consequently, the objective of the adversary is to minimize~\eqref{eq2}, while the agent aims at maximizing it. We refer to this particular dynamic cheap talk game as adversarial LQG (ALQG) game. Fig.~\ref{figure2} illustrates a three-stage ALQG game. Our objective is to characterize SPEs in the ALQG game: the existence conditions and solution structures.

\begin{remark}
The adversarial LQG problem cannot be modeled as an SG or a POSG, since the adversary directly manipulates the observation of the agent. Furthermore, different from zero-sum SG, an SPE of the ALQG game does not necessarily exist. On the other hand, for $N=1$ the game is a cheap talk game~\cite{crawford1982}, where the strategies of both players depend on the belief of the agent and on the constraints on the adversarial manipulation. Nonetheless, in the ALQG game the reward function is different from that in~\cite{crawford1982}, which gives rise to different equilibria, as we will show later.
\label{rm2}
\end{remark}

\section{EQUILIBRIUM ANALYSIS}
\label{sec4}
In the following, we first formulate the value function and the belief update rule for the adversarial LQG problem; we then characterize SPEs in pure strategies and in behavioral strategies, respectively.

\subsection{Value Function and Belief Update}
Assume that there is an SPE consisting of strategies $(g^{N*},f^{N*})$. Induced by this SPE, we can define the value function of a subgame starting from the $i$-th stage as
\begin{equation}
V_{i}^{N}(b_{i})=V(b_{i},g_{i}^{N*},f_{i}^{N*})=E_{b_{i},g_{i}^{N*},f_{i}^{N*}}\left(\sum_{j=i}^{N}R_{j}(S_{j},A_{j})\right),
\label{eq3}
\end{equation}
i.e., the value function $V_{i}^{N}(b_{i})$ is the expected accumulated agent reward in the subgame starting from the $i$-th stage when the belief in the $i$-th stage is $b_{i}$ and the SPE strategies $(g_{i}^{N*},f_{i}^{N*})$ are used. The evaluation of $V_{i}^{N}(b_{i})$ needs the beliefs in the subgame. In the following, we specify the belief update rule.

Given the current belief $b_{i}\triangleq\mathcal{N}(\mu_{i},\sigma_{i}^{2})$, the coefficient $\pi_{i}$, the variance $\delta_{i}^{2}$, the manipulated state $\hat{s}_{i}$, and the action $a_{i}$, it follows from the adversarial LQG model and Bayes rule that $b_{i+1}\triangleq\mathcal{N}(\mu_{i+1},\sigma_{i+1}^{2})$ with
\begin{gather}
\mu_{i+1}=\Lambda_{\mu}(b_{i},\pi_{i},\delta_{i}^{2},\hat{s}_{i},a_{i})=\alpha_{i}\frac{\pi_{i}\sigma_{i}^{2}\hat{s}_{i}+\mu_{i}\delta_{i}^{2}}{\pi_{i}^{2}\sigma_{i}^{2}+\delta_{i}^{2}}+\beta_{i}a_{i},\label{eq12}\\
\sigma_{i+1}^{2}=\Lambda_{\nu}(b_{i},\pi_{i},\delta_{i}^{2})=\frac{\alpha_{i}^{2}\sigma_{i}^{2}\delta_{i}^{2}}{\pi_{i}^{2}\sigma_{i}^{2}+\delta_{i}^{2}}+\omega_{i}^{2}.\label{eq13}
\end{gather}
An immediate consequence of the belief update rule is the following.
\begin{property}
It follows from $\sigma_{1}^2>0$, the variance update rule (\ref{eq13}), and $\pi_{i}\not=0$ that $\sigma_{i}^{2}>0$ for all $1\leq i\leq N$.
\label{prp1}
\end{property}

Observe that the value functions $\{V_{i}^N\}_{i=1}^{N-1}$ have to satisfy the backward dynamic programming equation
\begin{align}
V_{i}^{N}(b_{i})=&\,\min_{g_{i}}E_{b_{i},g_{i},f_{i}^{*}}\{R_{i}(S_{i},A_{i})\nonumber\\
&\,+V_{i+1}^{N}(\mathcal{N}(\Lambda_{\mu}(b_{i},\Pi_{i},\Delta_{i}^{2},\hat{S}_{i},A_{i}),\Lambda_{\nu}(b_{i},\Pi_{i},\Delta_{i}^{2})))\}\nonumber\\
=&\,\max_{f_{i}}E_{b_{i},g_{i}^{*},f_{i}}\{R_{i}(S_{i},A_{i})\nonumber\\
&\,+V_{i+1}^{N}(\mathcal{N}(\Lambda_{\mu}(b_{i},\Pi_{i},\Delta_{i}^{2},\hat{S}_{i},A_{i}),\Lambda_{\nu}(b_{i},\Pi_{i},\Delta_{i}^{2})))\}.
\label{eq5}
\end{align}
Thus the SPE has to satisfy~\eqref{eq5}, which is the basis for the analysis we present in the following.

\subsection{Pure Strategy Equilibria}
We start the analysis considering pure strategy equilibria. With slight abuse of notation, we denote by $(\pi_{i},\delta_{i}^{2})=g_{i}(b_{i})$ a pure strategy of the adversary as a function of the belief $b_{i}$.

We first consider the case $N=1$.

\begin{proposition}
Let $N=1$. An SPE consists of $(f_{1}^{*},g_{1}^{*})$, where $(\kappa_{1}^{*},\rho_{1}^{*})=f_{1}^{*}(b_{1})=(0,0)$ for any belief $b_{1}$; and $g_{1}^{*}$ can be any adversarial strategy defined on $\mathcal{A}_{1}(b_{1},\varepsilon',\varepsilon,\lambda)$.
\end{proposition}

\begin{proof}
Observe that $(\kappa_{1}^{*},\rho_{1}^{*})=f_{1}^{*}(b_{1})=(0,0)$ is a dominant strategy for the agent for any belief $b_{1}$. Under this strategy, the adversarial strategy has no impact on the agent's reward. This proves the result.
\end{proof}

The existence of a pure strategy SPE for $N=1$ is encouraging, even if the equilibrium is degenerate. Unfortunately, for $N\geq2$ an SPE may not exist as shown in the following theorem.

\begin{theorem}
Let $N\geq2$. If $\varepsilon'\not=\varepsilon$ or if $\varepsilon'=\varepsilon=0$, then there is no pure strategy SPE for the ALQG game. If $\varepsilon'=\varepsilon\not=0$, then there is a unique pure strategy SPE. The SPE strategies for $1\leq i\leq N$ are given by
\begin{gather}
\tilde{\theta}_{N+1}=\hat{\theta}_{N+1}=0;\label{eq21}\\
\tilde{\theta}_{i}=\theta_{i}+\tilde{\theta}_{i+1}\alpha_{i}^{2}-\frac{\tilde{\theta}_{i+1}^{2}\alpha_{i}^{2}\beta_{i}^{2}}{\phi_{i}+\tilde{\theta}_{i+1}\beta_{i}^{2}};\label{eq22}\\
\hat{\theta}_{i}=\theta_{i}+\hat{\theta}_{i+1}\alpha_{i}^{2}-\left(\frac{\tilde{\theta}_{i+1}^{2}\alpha_{i}^{2}\beta_{i}^{2}}{\phi_{i}+\tilde{\theta}_{i+1}\beta_{i}^{2}}+(\hat{\theta}_{i+1}-\tilde{\theta}_{i+1})\alpha_{i}^{2}\right)\frac{\lambda-1}{\lambda};\label{eq23}\\
\pi_{i}^{*}=g_{i}^{*}(b_{i})=\varepsilon'=\varepsilon;\\
\delta_{i}^{2*}=g_{i}^{*}(b_{i})=\frac{\varepsilon^{2}\sigma_{i}^{2}}{\lambda-1};\label{eq25}\\
\kappa_{i}^{*}=f_{i}^{*}(b_{i})=-\frac{\tilde{\theta}_{i+1}\alpha_{i}\beta_{i}(\lambda-1)}{(\phi_{i}+\tilde{\theta}_{i+1}\beta_{i}^{2})\lambda\varepsilon};\\
\rho_{i}^{*}=f_{i}^{*}(b_{i})=-\frac{\tilde{\theta}_{i+1}\alpha_{i}\beta_{i}\mu_{i}}{(\phi_{i}+\tilde{\theta}_{i+1}\beta_{i}^{2})\lambda}.
\end{gather}
\label{th1}
\end{theorem}

\begin{corollary}
If $\varepsilon'=\varepsilon\not=0$, the value function induced by the unique pure strategy SPE is
\begin{gather}
V_{i}^{N}(b_{i})=-\tilde{\theta}_{i}\mu_{i}^{2}-\hat{\theta}_{i}\sigma_{i}^{2}-\sum_{j=i+1}^{N}\hat{\theta}_{j}\omega_{j-1}^{2}.\label{eq28}
\end{gather}
\label{cor1}
\end{corollary}

The proofs of Theorem~\ref{th1} and Corollary~\ref{cor1} are provided in the appendix.

\begin{property}
It follows from (\ref{eq21})-(\ref{eq23}) that for $1\leq i\leq N$,
\begin{align}
\hat{\theta}_{i}\geq\tilde{\theta}_{i}>0.
\end{align}
\end{property}

\begin{remark}
We can make the following observations based on Theorem~\ref{th1} and Corollary~\ref{cor1}.
\begin{itemize}
\item Since the LQG control can be seen as the best response of any given (adversarial manipulated) observation model, it is sufficient to consider a pure agent strategy in the form of an affine function for an SPE of the ALQG game with a pure adversarial strategy.
\item It follows from~\eqref{eq25} that a rational adversary will always apply a manipulation with the largest variance.
\item The value function $V_{i}^{N}$ consists of a constant term and two separable terms depending on the mean $\mu_{i}$ and the variance $\sigma_{i}^{2}$ of the belief, which allows a closed form solution for arbitrary $N$.
\end{itemize}
\end{remark}

{\bf Time-invariant system:} We now turn to the asymptotic analysis of a time-invariant system, i.e., $\alpha_{i}=\alpha\not=0$, $\beta_{i}=\beta\not=0$, $\omega_{i}^{2}=\omega^{2}>0$, $\theta_{i}=\theta>0$, and $\phi_{i}=\phi>0$ for $i\geq1$, and we let $N\to\infty$. Let us define the mapping $L:\mathbb{R}^{2}_{\geq0}\to\mathbb{R}^{2}_{\geq0}$ as
\begin{align}
L\left(x,y\right)=\left(\theta+\frac{\phi\alpha^{2}x}{\phi+\beta ^{2}x},\theta+\frac{\phi\alpha^{2}x}{\phi+\beta ^{2}x}\frac{\lambda-1}{\lambda}+\alpha^{2}y\frac{1}{\lambda}\right).
\label{def:Lmapping}
\end{align}
Observe that $L$ is effectively the coefficient update (\ref{eq21})-(\ref{eq23}) for the time-invariant model. In what follows, we first characterize $L$ and furthermore the pure strategy SPE of the ALQG game in the asymptotic regime.

\begin{proposition}
Let $\lambda>\alpha^2$. Then the mapping $L$ admits a least fixed point $(\tilde{\theta},\hat{\theta})\in\mathbb{R}_{\geq0}^{2}$, for which
\begin{equation}
\lim_{n\to\infty}L^{n}(0,0)=L(\tilde{\theta},\hat{\theta})=(\tilde{\theta},\hat{\theta}),
\end{equation}
with $L^{n}(0,0)\triangleq \underbrace{L(L(\cdots(L(L}_{n\textnormal{ }L\textnormal{-mappings}}(0,0)))\cdots))$.
\label{pro2}
\end{proposition}

\begin{proof}
We start the proof by observing that the mapping $L$ is order-preserving. That is, for all $(x,y)$, $(x',y')\in\mathbb{R}^{2}_{\geq0}$ satisfying $(x,y)\preccurlyeq(x',y')$, i.e., $x\leq x'$ and $y\leq y'$, we have $L(x,y)\preccurlyeq L(x',y')$. This can be easily shown by analyzing~\eqref{def:Lmapping}.

Furthermore, the fixed point equation $L(\tilde{\theta},\hat{\theta})=(\tilde{\theta},\hat{\theta})$ has a unique solution on $\mathbb{R}_{\geq0}^{2}$ if and only if $\lambda>\alpha^2$. Since $L$ is order-preserving, the convergence result $\lim_{n\to\infty}L^{n}(0,0)=(\tilde{\theta},\hat{\theta})$ follows from Kleene's fixed point theorem~\cite{baranga1991}.
\end{proof}

Analytical expressions for $\tilde{\theta}$ and $\hat{\theta}$ can be obtained by solving the fixed point equation $L(\tilde{\theta},\hat{\theta})=(\tilde{\theta},\hat{\theta})$, and can be used for characterizing the SPE in pure strategies, using Theorem~\ref{th1} and Proposition~\ref{pro2}, as follows.

\begin{theorem}
Let $\lambda>\alpha^2$, $\varepsilon'=\varepsilon\not=0$, and $N\to\infty$. Then the ALQG game of the time-invariant model has a stationary SPE in pure strategies as: For $i\geq1$,
\begin{gather}
\pi_{i}^{*}=g_{i}^{*}(b_{i})=\varepsilon;\\
\delta_{i}^{2*}=g_{i}^{*}(b_{i})=\frac{\varepsilon^{2}\sigma_{i}^{2}}{\lambda-1};\\
\kappa_{i}^{*}=f_{i}^{*}(b_{i})=-\frac{\tilde{\theta}\alpha\beta(\lambda-1)}{(\phi+\tilde{\theta}\beta^{2})\lambda\varepsilon};\\
\rho_{i}^{*}=f_{i}^{*}(b_{i})=-\frac{\tilde{\theta}\alpha\beta\mu_{i}}{(\phi+\tilde{\theta}\beta^{2})\lambda}.
\end{gather}
\label{th2}
\end{theorem}

\begin{proof}
Since $\varepsilon'=\varepsilon\not=0$, a unique SPE in pure strategies exists and the SPE strategies are given in Theorem~\ref{th1}. Since $\lambda>\alpha^2$ and $N\to\infty$, it follows from Proposition~\ref{pro2} that $\tilde{\theta_{i}}$, $\hat{\theta}_{i}$ converge to $\tilde{\theta}$, $\hat{\theta}$, respectively. This leads to the stationary SPE in Theorem~\ref{th2}.
\end{proof}

Interestingly, for this stationary SPE in pure strategies we can obtain the expected average reward per stage in steady state in closed form.

\begin{corollary}
Let $b_{1}\triangleq\mathcal{N}(\mu_{1},\sigma_{1}^{2})$ with {\it bounded} mean and variance. For the stationary SPE in pure strategies in Theorem~\ref{th2}, the expected average reward per stage in steady state is independent of the initial belief and is given by
\begin{equation}
\lim_{N\to\infty}\frac{V_{1}^{N}(b_{1})}{N}=-\hat{\theta}\omega^{2}.
\end{equation}
\label{cor2}
\end{corollary}

\begin{proof}
This result follows from Corollary~\ref{cor1}, Proposition~\ref{pro2}, and Theorem~\ref{th2}.
\end{proof}

\subsection{Equilibria in Behavioral Strategies}
The previous results show that a pure strategy SPE does not exist if there are multiple choices of the coefficient $\pi_{i}$ for the adversary. We thus turn to the analysis of SPE in behavioral strategies.

\begin{theorem}
Let $N\geq2$ and $\varepsilon'<0<\varepsilon$. Then for $1\leq i\leq N$, $\tilde{\theta}_{i}$ and $\check{\theta}_{i}$ are as given by (\ref{eq22}), and
\begin{gather}
\tilde{\theta}_{N+1}=\check{\theta}_{N+1}=0,\label{eq37}\\
\check{\theta}_{i}=\theta_{i}+\check{\theta}_{i+1}\alpha_{i}^{2}-(\check{\theta}_{i+1}-\tilde{\theta}_{i+1})\alpha_{i}^{2}\frac{\lambda-1}{\lambda}.\label{eq38}
\end{gather}
Furthermore, there is a continuum of SPEs in behavioral strategies. Each SPE in the $i$-th stage consists of a behavioral strategy of the adversary and a pure strategy of the agent that satisfies
\begin{gather}
\mathbb{S}(g_{i}^{*}|b_{i})\triangleq\left\{(\pi_{i},\delta_{i}^{2}):\delta_{i}^{2}=\frac{\pi_{i}^{2}\sigma_{i}^{2}}{\lambda-1}\right\}\subseteq \mathcal{A}_{i}(b_{i},\varepsilon',\varepsilon,\lambda);\\
||\mathbb{S}(g_{i}^{*}|b_{i})||\geq2;\\
E_{g_{i}^{*}}(\Pi_{i})=0;\\
(\kappa_{i}^{*},\rho_{i}^{*})=f_{i}^{*}(b_{i})=\left(0,-\frac{\tilde{\theta}_{i+1}\alpha_{i}\beta_{i}\mu_{i}}{\phi_{i}+\tilde{\theta}_{i+1}\beta_{i}^{2}}\right).
\end{gather}
\label{th3}
\end{theorem}

Interestingly, this SPE is a babbling equilibrium in which the agent's action is based on its belief, not on the manipulated state. Nonetheless, the value of the game depends on the adversarial manipulation, as shown in the following corollary.

\begin{corollary}
Let $\varepsilon'<0<\varepsilon$. For any SPE in behavioral strategies, we have
\begin{equation}
V_{i}^{N}(b_{i})=-\tilde{\theta}_{i}\mu_{i}^{2}-\check{\theta}_{i}\sigma_{i}^{2}-\sum_{j=i+1}^{N}\check{\theta}_{j}\omega_{j-1}^{2}.
\label{eq41}
\end{equation}
\label{cor3}
\end{corollary}

The proofs of Theorem~\ref{th3} and Corollary~\ref{cor3} are given in the appendix.


\begin{property}
It follows from the update rules (\ref{eq21})-(\ref{eq23}) and (\ref{eq37})-(\ref{eq38}) that for all $1\leq i\leq N$,
\begin{align}
\check{\theta}_{i}\geq\hat{\theta}_{i}\geq\tilde{\theta}_{i}>0.
\end{align}
\label{prp4}
\end{property}

\begin{remark}
We can make the following observations based on Theorem~\ref{th3} and Corollary~\ref{cor3}.
\begin{itemize}
\item It is sufficient for the agent to use a pure affine strategy against a behavioral strategy of the adversary.
\item Although the adversary cannot use $\pi_{i}=0$, the behavioral strategy $g_{i}^{*}$ needs to achieve zero-mean of the random coefficient $\Pi_{i}$.
\item A rational adversary will always use a manipulation with the largest variance.
\item From Property~\ref{prp4}, the value (\ref{eq41}) of an SPE in behavioral strategies is always less than or equal to the value (\ref{eq28}) of a pure strategy SPE.
\end{itemize}
\end{remark}

{\bf Time-invariant system:} We again turn to the time-invariant system for $N\to\infty$. Let us define the mapping $J:\mathbb{R}_{\geq0}^{2}\to\mathbb{R}_{\geq0}^{2}$ as
\begin{align}
J(x,y)=\left(\theta+\frac{\phi\alpha^{2}x}{\phi+\beta^{2}x},\theta+\alpha^{2}x\frac{\lambda-1}{\lambda}+\alpha^{2}y\frac{1}{\lambda}\right).
\end{align}
Observe that $J$ is effectively the coefficient update rule (\ref{eq22}), (\ref{eq37}), (\ref{eq38}) for the time-invariant system. In what follows, we characterize $J$ and the stationary SPEs in behavioral strategies for the ALQG game.

\begin{proposition}
Let $\lambda>\alpha^2$. Then the mapping $J$ admits a least fixed point $(\tilde{\theta},\check{\theta})\in\mathbb{R}_{\geq0}^{2}$, for which
\begin{align}
\lim_{n\to\infty}J^{n}(0,0)=J(\tilde{\theta},\check{\theta})=(\tilde{\theta},\check{\theta}).
\end{align}
\label{pro3}
\end{proposition}

The proof of Proposition~\ref{pro3} follows using the arguments in the proof of Proposition~\ref{pro2}.

\begin{theorem}
Let $\lambda>\alpha^2$, $\varepsilon'<0<\varepsilon$, and $N\to\infty$. Then the ALQG game of the time-invariant model has a stationary SPE in behavioral strategies as: For $i\geq1$,
\begin{gather}
g_{i}^{*}\left(\pi_{i}=\varepsilon',\left.\delta_{i}^{2}=\frac{\varepsilon'^{2}\sigma_{i}^{2}}{\lambda-1}\right|b_{i}\right)=\frac{\varepsilon}{\varepsilon-\varepsilon'};\\
g_{i}^{*}\left(\pi_{i}=\varepsilon,\left.\delta_{i}^{2}=\frac{\varepsilon^{2}\sigma_{i}^{2}}{\lambda-1}\right|b_{i}\right)=-\frac{\varepsilon'}{\varepsilon-\varepsilon'};\\
(\kappa_{i}^{*},\rho_{i}^{*})=f_{i}^{*}(b_{i})=\left(0,-\frac{\tilde{\theta}\alpha\beta\mu_{i}}{\phi+\tilde{\theta}\beta^{2}}\right).
\end{gather}
\label{th4}
\end{theorem}

\begin{corollary}
Let $b_{1}\triangleq\mathcal{N}(\mu_{1},\sigma_{1}^{2})$ with {\it bounded} mean and variance. For the stationary SPE in behavioral strategies in Theorem~\ref{th4}, the expected average reward per stage in steady state is independent of the initial belief and is given by
\begin{equation}
\lim_{N\to\infty}\frac{V_{1}^{N}(b_{1})}{N}=-\check{\theta}\omega^{2}.
\end{equation}
\label{cor4}
\end{corollary}

The proofs of Theorem~\ref{th4} and Corollary~\ref{cor4} are based on Theorem~\ref{th3}, Corollary~\ref{cor3}, and Proposition~\ref{pro3}, and follow from similar arguments as used in the proofs of Theorem~\ref{th2} and Corollary~\ref{cor2}. Observe that Corollary~\ref{cor2}, Corollary~\ref{cor4}, and Property~\ref{prp4} jointly imply that the expected average agent reward per stage in steady state is higher when considering pure strategies, as behavioral strategies allow for more uncertainty about the attack and thus make the adversary stronger.

The behavioral strategy of the adversary in Theorem~\ref{th3} needs to achieve zero-mean of the random coefficient $\Pi_{i}$, which {\it cannot} be satisfied if $0\leq\varepsilon'<\varepsilon$ or $\varepsilon'<\varepsilon\leq0$. In the following, we study SPE under these conditions.

\begin{theorem}
Let $N\geq2$. If $0=\varepsilon'<\varepsilon$ or if $\varepsilon'<\varepsilon=0$, there is no SPE for the ALQG game.
\label{th5}
\end{theorem}

\begin{theorem}
Let $N=2$. If $0<\varepsilon'<\varepsilon$ or if $\varepsilon'<\varepsilon<0$, there is a unique SPE in behavioral strategies for the ALQG game: For any belief $b_{1}\triangleq\mathcal{N}(\mu_{1},\sigma_{1}^{2})$,
\begin{gather}
g_{1}^{*}\left(\pi_{1}=\varepsilon',\left.\delta_{1}^{2}=\frac{\varepsilon'^{2}\sigma_{1}^{2}}{\lambda-1}\right|b_{1}\right)=\frac{\varepsilon}{\varepsilon'+\varepsilon};\\
g_{1}^{*}\left(\pi_{1}=\varepsilon,\left.\delta_{1}^{2}=\frac{\varepsilon^{2}\sigma_{1}^{2}}{\lambda-1}\right|b_{1}\right)=\frac{\varepsilon'}{\varepsilon'+\varepsilon};\\
\kappa_{1}^{*}=f_{1}^{*}(b_{1})\qquad\qquad\qquad\qquad\qquad\qquad\qquad\qquad\qquad\qquad\qquad\qquad\\
\;\;\;=\frac{-\theta_{2}\alpha_{1}\beta_{1}E_{g_{1}^{*}}(\Pi_{1})\sigma_{1}^{2}}{(\phi_{1}+\theta_{2}\beta_{1}^{2})\left(E_{g_{1}^{*}}(\Pi_{1}^{2})\left(\mu_{1}^{2}+\frac{\lambda}{\lambda-1}\sigma_{1}^{2}\right)-E_{g_{1}^{*}}^{2}(\Pi_{1})\mu_{1}^{2}\right)};\\
\rho_{1}^{*}=f_{1}^{*}(b_{1})=-E_{g_{1}^{*}}(\Pi_{1})\mu_{1}\kappa_{1}^{*}-\frac{\theta_{2}\alpha_{1}\beta_{1}\mu_{1}}{\phi_{1}+\theta_{2}\beta_{1}^{2}};
\end{gather}
for any belief $b_{2}\triangleq\mathcal{N}(\mu_{2},\sigma_{2}^{2})$, $g_{2}^{*}=g_{1}^{*}$; $(\kappa_{2}^{*},\rho_{2}^{*})=f_{2}^{*}(b_{2})=(0,0)$; and
\begin{align}
V_{1}^{2}&(b_{1})\nonumber\\
=&\,-\left(\theta_{1}+\theta_{2}\alpha_{1}^{2}-\frac{\theta_{2}^{2}\alpha_{1}^{2}\beta_{1}^{2}}{\phi_{1}+\theta_{2}\beta_{1}^{2}}\right)\mu_{1}^{2}-\theta_{2}\omega_{1}^{2}-(\theta_{1}+\theta_{2}\alpha_{1}^{2})\sigma_{1}^{2}\nonumber\\
&\,+\frac{\theta_{2}^{2}\alpha_{1}^{2}\beta_{1}^{2}E_{g_{1}^{*}}^{2}(\Pi_{1})\sigma_{1}^{4}}{(\phi_{1}+\theta_{2}\beta_{1}^{2})\left(E_{g_{1}^{*}}(\Pi_{1}^{2})\left(\mu_{1}^{2}+\frac{\lambda}{\lambda-1}\sigma_{1}^{2}\right)-E_{g_{1}^{*}}^{2}(\Pi_{1})\mu_{1}^{2}\right)}.
\label{eq54}
\end{align}
\label{th6}
\end{theorem}

The proofs of Theorems~\ref{th5}~\&~\ref{th6} are given in the appendix.


\begin{remark}
Different from the value functions (\ref{eq28}) and (\ref{eq41}), the value function (\ref{eq54}) {\it cannot} be decomposed into separable terms of mean and variance, which makes the extension of Theorem~\ref{th6} to $N\geq3$ difficult.
\end{remark}

\section{NUMERICAL RESULTS}
\label{sec5}
We illustrate the impact of strategic interaction on a time-invariant adversarial LQG problem. The parameters used for the evaluation are shown in Table~\ref{tab1}.

\begin{table}[b]
\caption{LQG Model Default Parameters}
\label{tab1}
\begin{center}
\begin{tabular}{|c||c|c|c|c|c|c|c|}
\hline
Parameter & $\mu_{1}$ & $\sigma_{1}^{2}$ & $\alpha$ & $\beta$ & $\omega^{2}$ & $\theta$ & $\phi$\\
\hline
Value & $0$ & $1$ & $-0.5$ & $-1.5$ & $1$ & $2$ & $1$\\
\hline
\end{tabular}
\end{center}
\end{table}

\begin{figure*}[t!]
\begin{minipage}[t]{0.3\textwidth}
\begin{center}
\centerline{\includegraphics[scale=0.35]{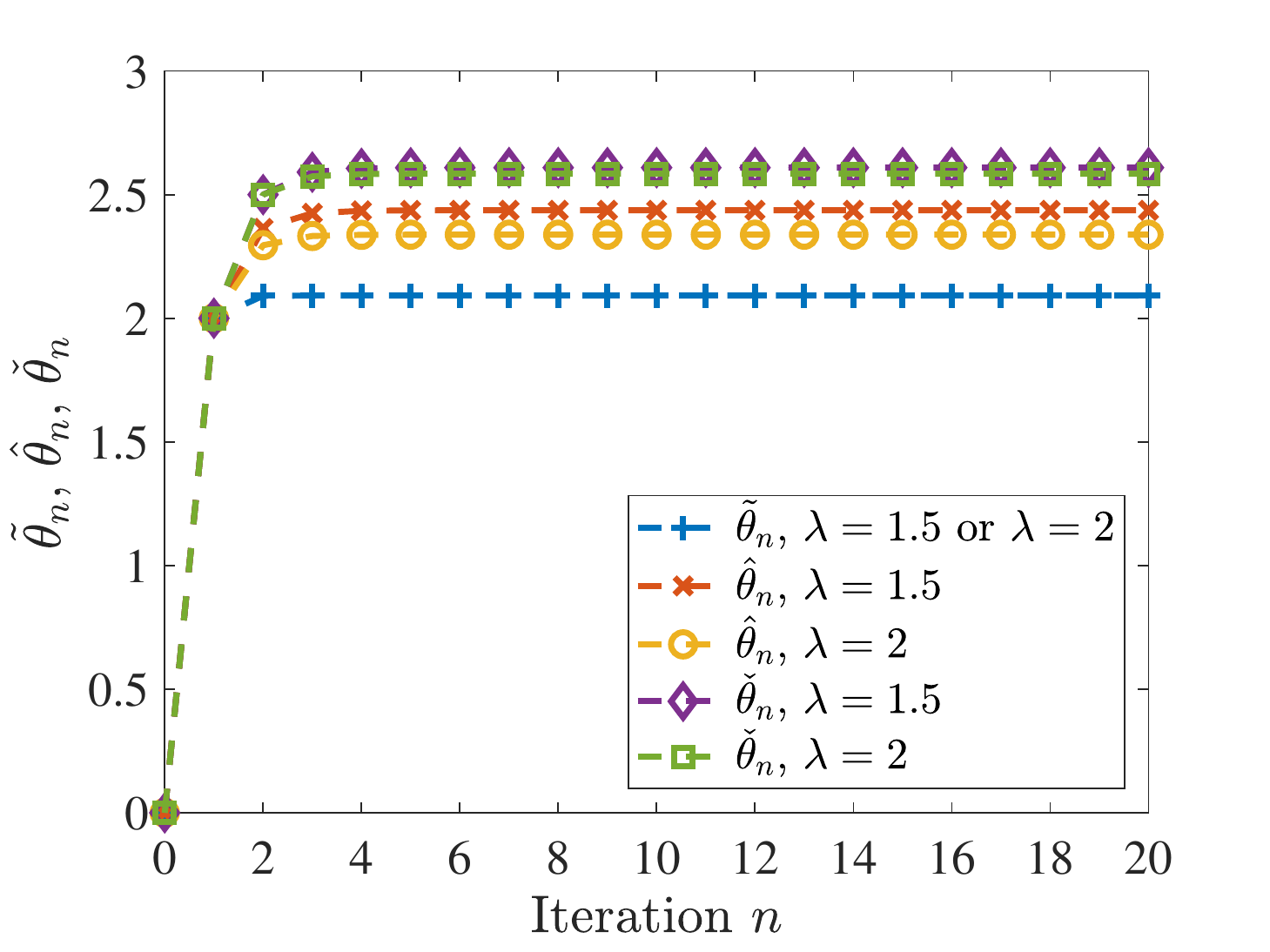}}
\caption{$\tilde{\theta}_{n}$, $\hat{\theta}_{n}$, $\check{\theta}_{n}$ computed as $L^{n}(0,0)$ and $J^{n}(0,0)$ v.s. the number of iterations $n$, for $\lambda=1.5$ and $\lambda=2$ ($\lambda>\alpha^{2}$).}
\label{figure4}
\end{center}
\end{minipage}
\hfill
\begin{minipage}[t]{0.3\textwidth}
\begin{center}
\centerline{\includegraphics[scale=0.35]{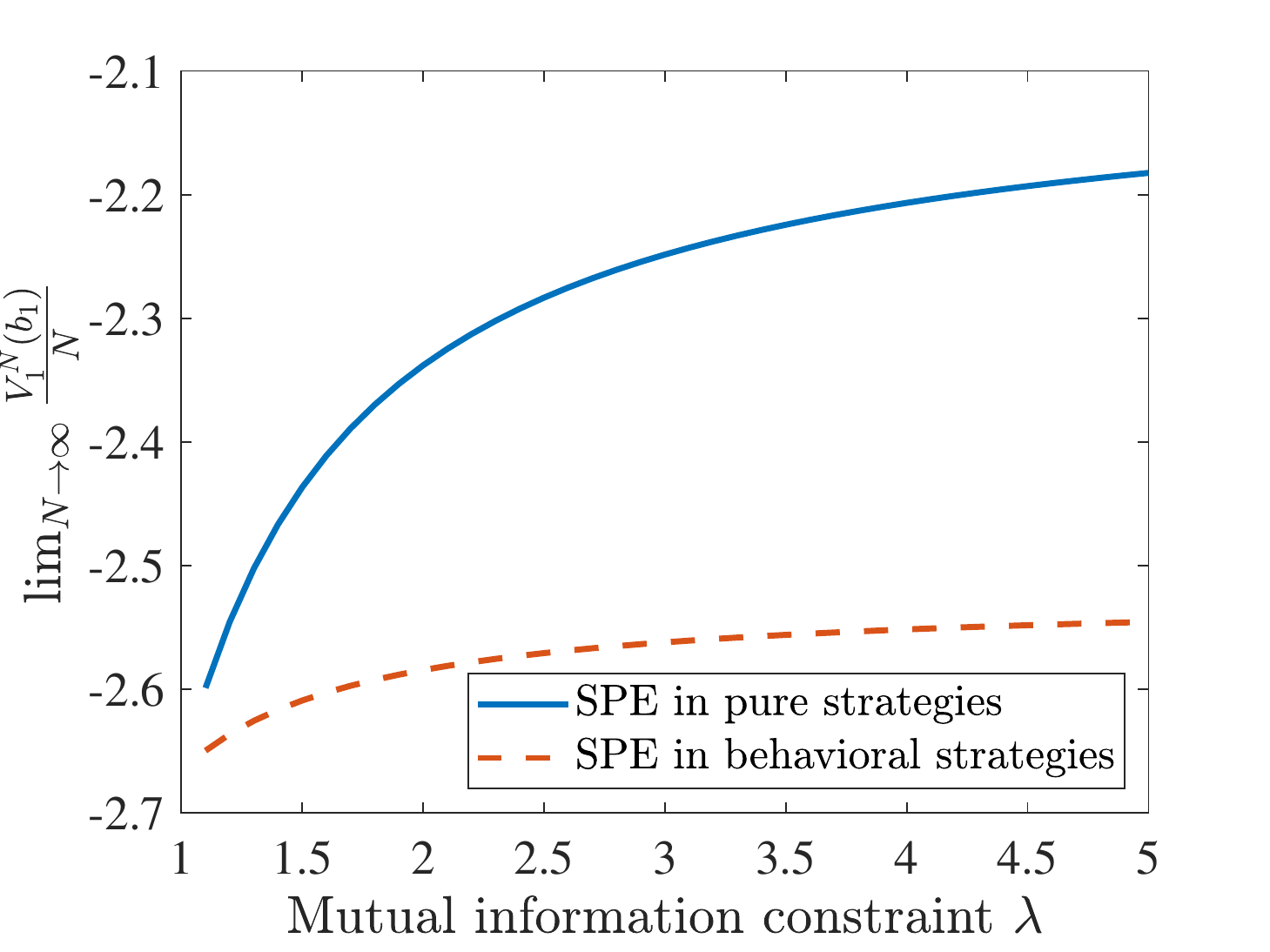}}
\caption{Expected average reward per stage v.s. mutual information constraint $\lambda$, for stationary SPEs in pure strategies and behavioral strategies.}
\label{figure3}
\end{center}
\end{minipage}
\hfill
\begin{minipage}[t]{0.3\textwidth}
\begin{center}
\centerline{\includegraphics[scale=0.35]{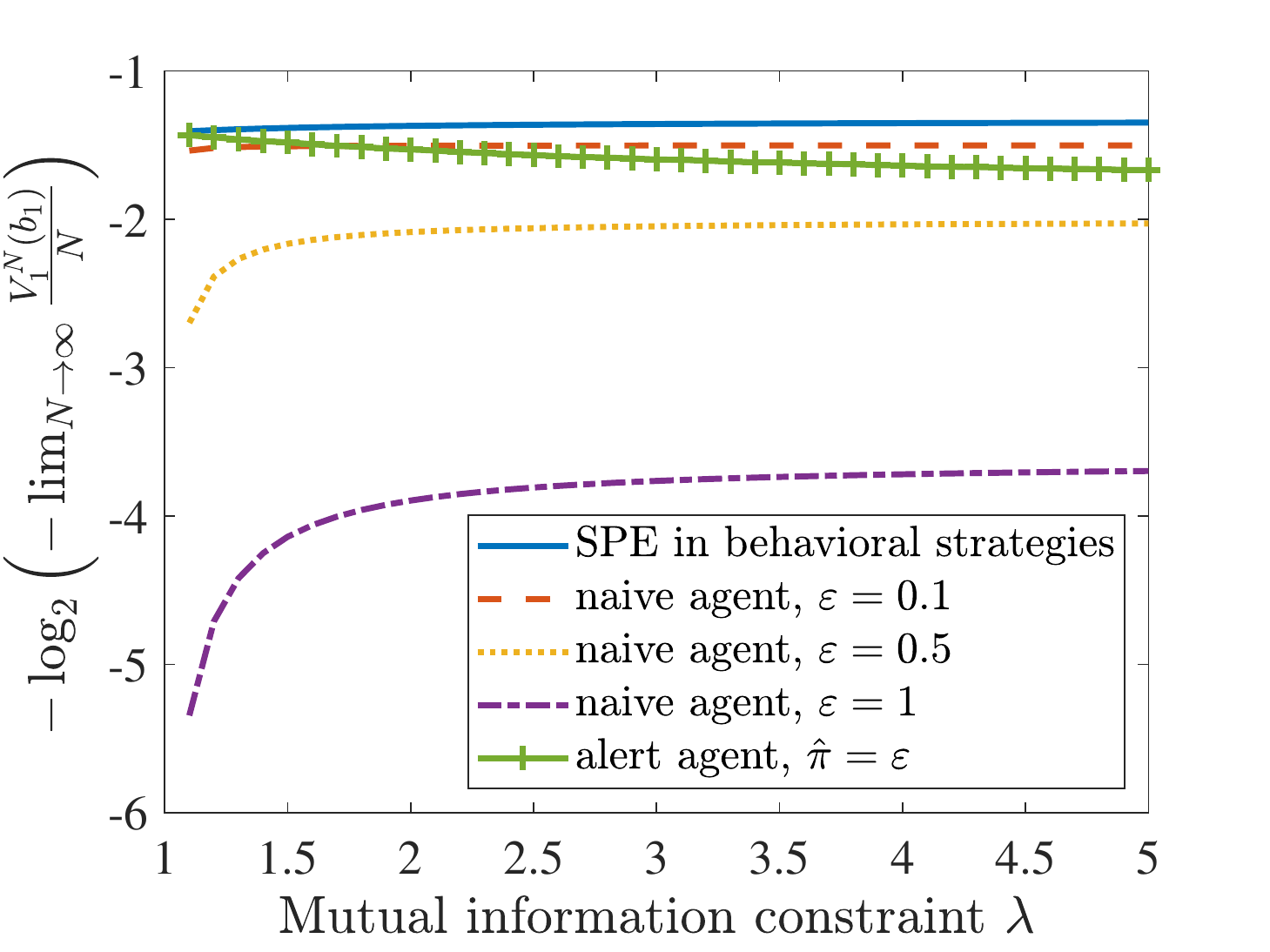}}
\caption{Expected average reward per stage for stationary SPE in behavioral strategies, that for a naive agent, and that for an alert agent v.s. mutual information constraint $\lambda$.}
\label{figure6}
\end{center}
\end{minipage}
\end{figure*}

We start with illustrating the convergence of the mappings in Propositions~\ref{pro2} and~\ref{pro3}. Fig.~\ref{figure4} shows the mappings $L^{n}(0,0)$ and $J^{n}(0,0)$ as functions of the iteration $n$ for $\lambda=1.5$ and $\lambda=2$. Both mappings increase monotonically starting from $(0,0)$ and converge to their least fixed points $(\tilde{\theta},\hat{\theta})$ and $(\tilde{\theta},\check{\theta})$, respectively, which confirms these propositions.

We continue with the evaluation of stationary SPEs for the time-invariant system. Fig.~\ref{figure3} shows the expected average reward per stage for a stationary SPE in pure strategies and that for a stationary SPE in behavioral strategies. Observe that the adversarial manipulation capability decreases as the mutual information lower bound $\lambda$ increases, and hence the expected average rewards increase. The results also confirm Property~\ref{prp4}, i.e., the expected average reward per stage for a stationary SPE in behavioral strategies {\it cannot} be higher than that for a stationary SPE in pure strategies.

Next, we assess the importance of strategic interaction on the agent's performance, as compared to an agent that is unaware of the attack~\cite{huang2016,russo2019}. Fig.~\ref{figure6} shows the expected average reward per stage for a stationary SPE in behavioral strategies, as per Corollary~\ref{cor4}, that for a naive agent under an optimal adversarial attack, and that for an alert agent under SPE adversarial behavioral strategy $g_{i}^{*}$ in Corollary~\ref{cor4}, for $-\varepsilon'=\varepsilon>0$. The naive agent is unaware of the adversarial manipulation, i.e., it uses the optimal LQR strategy (\ref{eq:LQR}). The corresponding optimal stationary adversarial strategy can be obtained through dynamic programming, and is the pure strategy:
\begin{align*}
(\pi_{i},\delta_{i}^{2})=g^A_{i}(b_{i})=\left(-\varepsilon,\frac{\varepsilon^{2}\sigma_{i}^{2}}{\lambda-1}\right).
\end{align*}
The alert agent suspects an adversary but does not act strategically despite the presence of an adversary. The alert agent assumes $\hat{\pi}_{i}=\hat{\pi}\not=0$, $\hat{\delta}_{i}^{2}=\frac{\hat{\pi}^{2}\sigma_{i}^{2}}{\lambda-1}$, and uses the corresponding best response strategy
\begin{align*}
(\kappa_{i},\rho_{i})=f_{i}^{C}(b_{i})=\left(-\frac{\tilde{\theta}\alpha\beta(\lambda-1)}{(\phi+\tilde{\theta}\beta^{2})\lambda\hat{\pi}},-\frac{\tilde{\theta}\alpha\beta\mu_{i}}{(\phi+\tilde{\theta}\beta^{2})\lambda}\right).
\end{align*}
The figure shows that the expected average reward per stage of the naive agent is always lower than that for the SPE. Clearly, as the mutual information lower bound $\lambda$ increases, the adversarial manipulation capability becomes weaker and the expected average rewards per stage increase. At the same time, we can observe that if the bound $\varepsilon$ of the manipulation coefficient is higher then the adversarial manipulation capability becomes stronger and therefore the expected average reward per stage of the naive agent decreases. Note that the limit of the expected average reward per stage of the naive agent does not exist when $\varepsilon$ is larger than a threshold due to the resulting instability of the control system. The poor performance of the naive agent is consistent with recent works on adversarial DRL~\cite{huang2016,russo2019}, where the naive DRL agents were found to perform poorly against strategic adversaries. The results for the SPE show, however, that an agent that is aware of the adversary can adjust its strategy to be resilient to adversarial attack. The figure also shows that the expected average reward per stage of the alert agent is always lower than that for the SPE since the alert agent does not adjust its best response strategically to the SPE adversarial strategy. Different from the SPE and the naive agent, it is interesting to observe that the alert agent's performance deteriorates as the adversarial constraint $\lambda$ increases. This is due to that the alert agent's strategy $f_{i}^{C}$ deviates more from the SPE strategy $f_{i}^{*}$, as shown in Corollary~\ref{cor4}, as $\lambda$ increases.



\section{CONCLUSION}
\label{sec6}
We proposed a game theoretic model to capture the strategic interaction, information asymmetry and system dynamics for LQG control under adversarial input subject to a mutual information constraint. We characterized the subgame perfect equilibria in pure strategies and in behavioral strategies, including stationary equilibria for time-invariant systems. Our results show that if an equilibrium exists then the agent can use an affine, pure strategy, but randomization enables the adversary to construct more powerful attacks, under a wider range of parameters, and forces the agent into a babbling equilibrium. Our numerical results show the importance of strategic interaction for LQG control, and highlight that an agent that is aware of an adversarial attack can be designed resilient. Our work could be extended in a number of interesting directions, including considering a non-scalar state dynamic system, and relaxing the assumption that the adversarial strategy is revealed to the agent after each stage. 

%

\appendix
\subsection{Proofs of Theorem~\ref{th1} and Corollary~\ref{cor1}}
\begin{proof}
We prove the result using backward dynamic programming. Recall that there is no feasible adversarial strategy for $\varepsilon'=\varepsilon=0$. Thus, it is sufficient to consider the cases
$\varepsilon'\not=\varepsilon$ or $\varepsilon'=\varepsilon\not=0$. 
In stage $N$ the value function for $b_{N}$ can be expressed as
\begin{align}
V_{N}^{N}(b_{N})=-\tilde{\theta}_{N}\mu_{N}^{2}-\hat{\theta}_{N}\sigma_{N}^{2},\label{a1}
\end{align}
where $\tilde{\theta}_{N}=\hat{\theta}_{N}=\theta_{N}$ from the update rules (\ref{eq21})-(\ref{eq23}). The pure strategies, which form an SPE and achieve the value function, consist of any pure strategy $g_{N}^{*}$ satisfying the adversarial constraints (\ref{eq14})-(\ref{eq15}), and the dominant pure strategy
\begin{align}
(\kappa_{N}^{*},\rho_{N}^{*})=f_{N}^{*}(b_{N})=(0,0).\label{a2}
\end{align}

In stage $N-1$, the Q-function of using $\pi_{N-1}\not=0$, $\delta_{N-1}^{2}$, $\kappa_{N-1}$, and $\rho_{N-1}$ given a belief $b_{N-1}$ is
\begin{align}
Q_{N-1}^{N}&(b_{N-1},\pi_{N-1},\delta_{N-1}^{2},\kappa_{N-1},\rho_{N-1})\nonumber\\
=&\,-\theta_{N-1}E(S_{N-1}^{2})-\phi_{N-1}E(A_{N-1}^{2})\nonumber\\
&\,-\tilde{\theta}_{N}E(\Lambda_{\mu}(b_{N-1},\pi_{N-1},\delta_{N-1}^{2},\hat{S}_{N-1},A_{N-1}))^{2}\nonumber\\
&\,-\hat{\theta}_{N}\Lambda_{\nu}(b_{N-1},\pi_{N-1},\delta_{N-1}^{2})\nonumber\\
=&\,-(\theta_{N-1}+\tilde{\theta}_{N}\alpha_{N-1}^{2})\mu_{N-1}^{2}-\hat{\theta}_{N}\omega_{N-1}^{2}\nonumber\\
&\,-(\theta_{N-1}+\hat{\theta}_{N}\alpha_{N-1}^{2})\sigma_{N-1}^{2}\nonumber\\
&\,-(\phi_{N-1}+\tilde{\theta}_{N}\beta_{N-1}^{2})(\pi_{N-1}\kappa_{N-1}\mu_{N-1}+\rho_{N-1})^{2}\nonumber\\
&\,-2\tilde{\theta}_{N}\alpha_{N-1}\beta_{N-1}\mu_{N-1}(\pi_{N-1}\kappa_{N-1}\mu_{N-1}+\rho_{N-1})\nonumber\\
&\,+(\hat{\theta}_{N}-\tilde{\theta}_{N})\alpha_{N-1}^{2}\frac{\pi_{N-1}^{2}\sigma_{N-1}^{2}}{\pi_{N-1}^{2}\sigma_{N-1}^{2}+\delta_{N-1}^{2}}\sigma_{N-1}^{2}\nonumber\\
&\,-(\phi_{N-1}+\tilde{\theta}_{N}\beta_{N-1}^{2})\kappa_{N-1}^{2}(\pi_{N-1}^{2}\sigma_{N-1}^{2}+\delta_{N-1}^{2})\nonumber\\
&\,-2\tilde{\theta}_{N}\alpha_{N-1}\beta_{N-1}\pi_{N-1}\kappa_{N-1}\sigma_{N-1}^{2},\label{a3}
\end{align}
where the expectations are induced by the given $b_{N-1}$, $\pi_{N-1}$, $\delta_{N-1}^{2}$, $\kappa_{N-1}$, and $\rho_{N-1}$.

Given $b_{N-1}$, $\pi_{N-1}\not=0$, $\delta_{N-1}^{2}$, and $\kappa_{N-1}$, the Q-function $Q_{N-1}^{N}$ is a concave quadratic function of $\rho_{N-1}$. As the best response to maximize the agent reward, we can substitute $\rho_{N-1}$ in terms of $b_{N-1}$, $\pi_{N-1}$, and $\kappa_{N-1}$ as
\begin{align}
\rho_{N-1}=-\pi_{N-1}\kappa_{N-1}\mu_{N-1}-\frac{\tilde{\theta}_{N}\alpha_{N-1}\beta_{N-1}}{\phi_{N-1}+\tilde{\theta}_{N}\beta_{N-1}^{2}}\mu_{N-1}.\label{a4}
\end{align}

Thus, it is sufficient to consider the Q-function
\begin{align}
Q_{N-1}^{N}&(b_{N-1},\pi_{N-1},\delta_{N-1}^{2},\kappa_{N-1})\nonumber\\
=&\,-\left(\theta_{N-1}+\tilde{\theta}_{N}\alpha_{N-1}^{2}-\frac{\tilde{\theta}_{N}^{2}\alpha_{N-1}^{2}\beta_{N-1}^{2}}{\phi_{N-1}+\tilde{\theta}_{N}\beta_{N-1}^{2}}\right)\mu_{N-1}^{2}\nonumber\\
&\,-(\theta_{N-1}+\hat{\theta}_{N}\alpha_{N-1}^{2})\sigma_{N-1}^{2}-\hat{\theta}_{N}\omega_{N-1}^{2}\nonumber\\
&\,+(\hat{\theta}_{N}-\tilde{\theta}_{N})\alpha_{N-1}^{2}\frac{\pi_{N-1}^{2}\sigma_{N-1}^{2}}{\pi_{N-1}^{2}\sigma_{N-1}^{2}+\delta_{N-1}^{2}}\sigma_{N-1}^{2}\nonumber\\
&\,-(\phi_{N-1}+\tilde{\theta}_{N}\beta_{N-1}^{2})\kappa_{N-1}^{2}(\pi_{N-1}^{2}\sigma_{N-1}^{2}+\delta_{N-1}^{2})\nonumber\\
&\,-2\tilde{\theta}_{N}\alpha_{N-1}\beta_{N-1}\pi_{N-1}\kappa_{N-1}\sigma_{N-1}^{2}.\label{a5}
\end{align}

As shown in Property~\ref{prp1}, the belief variance is always positive. Therefore, given $b_{N-1}$, $\pi_{N-1}\not=0$, and $\delta_{N-1}^{2}$, the Q-function $Q_{N-1}^{N}$ is a concave quadratic function of $\kappa_{N-1}$, and the best response of the agent in terms of $\kappa_{N-1}$ for $b_{N-1}$, $\pi_{N-1}$, and $\delta_{N-1}^{2}$ can be expressed as 
\begin{align}
\kappa_{N-1}=-\frac{\tilde{\theta}_{N}\alpha_{N-1}\beta_{N-1}\pi_{N-1}\sigma_{N-1}^{2}}{(\phi_{N-1}+\tilde{\theta}_{N}\beta_{N-1}^{2})(\pi_{N-1}^{2}\sigma_{N-1}^{2}+\delta_{N-1}^{2})}\not=0.\label{a6}
\end{align}

Given $b_{N-1}$, $\pi_{N-1}\not=0$, and $\kappa_{N-1}\not=0$, the Q-function $Q_{N-1}^{N}$ is a decreasing function of $\delta_{N-1}^{2}$. The best response of the adversary in terms of $\delta_{N-1}^{2}$ for $b_{N-1}$ and $\pi_{N-1}$ is thus
\begin{align}
\delta_{N-1}^{2}=\frac{\pi_{N-1}^{2}\sigma_{N-1}^{2}}{\lambda-1}.\label{a7}
\end{align}

Consequently, it is sufficient to consider the Q-function
\begin{align}
Q_{N-1}^{N}&(b_{N-1},\pi_{N-1},\kappa_{N-1})\nonumber\\
=&\,-\left(\theta_{N-1}+\tilde{\theta}_{N}\alpha_{N-1}^{2}-\frac{\tilde{\theta}_{N}^{2}\alpha_{N-1}^{2}\beta_{N-1}^{2}}{\phi_{N-1}+\tilde{\theta}_{N}\beta_{N-1}^{2}}\right)\mu_{N-1}^{2}\nonumber\\
&\,-(\theta_{N-1}+\hat{\theta}_{N}\alpha_{N-1}^{2})\sigma_{N-1}^{2}-\hat{\theta}_{N}\omega_{N-1}^{2}\nonumber\\
&\,+(\hat{\theta}_{N}-\tilde{\theta}_{N})\alpha_{N-1}^{2}\frac{\lambda-1}{\lambda}\sigma_{N-1}^{2}\nonumber\\
&\,-(\phi_{N-1}+\tilde{\theta}_{N}\beta_{N-1}^{2})\kappa_{N-1}^{2}\frac{\lambda}{\lambda-1}\pi_{N-1}^{2}\sigma_{N-1}^{2}\nonumber\\
&\,-2\tilde{\theta}_{N}\alpha_{N-1}\beta_{N-1}\pi_{N-1}\kappa_{N-1}\sigma_{N-1}^{2}.\label{a8}
\end{align}
The pure strategies $(g_{N-1}^{*},f_{N-1}^{*})$ form an SPE if $\pi_{N-1}^{*}=g_{N-1}^{*}(b_{N-1})$ and $\kappa_{N-1}^{*}=f_{N-1}^{*}(b_{N-1})$ satisfy
\begin{gather}
\pi_{N-1}^{*}=\arg\min_{\varepsilon'\leq\pi_{N-1}\leq\varepsilon,\pi_{N-1}\not=0}Q_{N-1}^{N}(b_{N-1},\pi_{N-1},\kappa_{N-1}^{*}),\label{a9}\\
\kappa_{N-1}^{*}=\arg\max_{\kappa_{N-1}\in\mathbb{R}}Q_{N-1}^{N}(b_{N-1},\pi_{N-1}^{*},\kappa_{N-1}).\label{a10}
\end{gather}
If $\varepsilon'=\varepsilon\not=0$, $\pi_{N-1}^{*}=\varepsilon=\varepsilon'=g_{N-1}^{*}(b_{N-1})$ is a dominant adversarial strategy. Therefore, the SPE must exist. The pure strategies $(g_{N-1}^{*},f_{N-1}^{*})$ can be obtained by substituting $\pi_{N-1}^{*}=\varepsilon$ into (\ref{a4}), (\ref{a6}), and (\ref{a7}) as
\begin{gather*}
\delta_{N-1}^{2*}=g_{N-1}^{*}(b_{N-1})=\frac{\varepsilon^{2}\sigma_{N-1}^{2}}{\lambda-1};\\
\kappa_{N-1}^{*}=f_{N-1}^{*}(b_{N-1})=-\frac{\tilde{\theta}_{N}\alpha_{N-1}\beta_{N-1}(\lambda-1)}{(\phi_{N-1}+\tilde{\theta}_{N}\beta_{N-1}^{2})\lambda\varepsilon};\\
\rho_{N-1}^{*}=f_{N-1}^{*}(b_{N-1})=-\frac{\tilde{\theta}_{N}\alpha_{N-1}\beta_{N-1}\mu_{N-1}}{(\phi_{N-1}+\tilde{\theta}_{N}\beta_{N-1}^{2})\lambda}.
\end{gather*}
The value function $V_{N-1}^{N}(b_{N-1})$ can be obtained by substituting $\pi_{N-1}^{*}$ and $\kappa_{N-1}^{*}$ into the Q-function (\ref{a8}) as
\begin{gather*}
V_{N-1}^{N}(b_{N-1})=-\tilde{\theta}_{N-1}\mu_{N-1}^{2}-\hat{\theta}_{N-1}\sigma_{N-1}^{2}-\hat{\theta}_{N}\omega_{N-1}^{2}.
\end{gather*}
Let us now consider the case $\varepsilon'\not=\varepsilon$. Assume that there exists an SPE with $\pi_{N-1}^{*}=g_{N-1}^{*}(b_{N-1})\not=0$. As the best response, solving (\ref{a10}) leads to $\kappa_{N-1}^{*}=-\frac{\tilde{\theta}_{N}\alpha_{N-1}\beta_{N-1}(\lambda-1)}{(\phi_{N-1}+\tilde{\theta}_{N}\beta_{N-1}^{2})\lambda\pi_{N-1}^{*}}$. For all $\pi_{N-1}\not=\pi_{N-1}^{*}$ and $\pi_{N-1}\not=0$ we have
\begin{align*}
Q_{N-1}^{N}(b_{N-1},\pi_{N-1}^{*},\kappa_{N-1}^{*})>Q_{N-1}^{N}(b_{N-1},\pi_{N-1},\kappa_{N-1}^{*}).
\end{align*}
Thus, condition (\ref{a9}) cannot hold and hence the assumption is not true, i.e., there is no pure strategy SPE in this case.

In the case of $\varepsilon'=\varepsilon\not=0$, Theorem~\ref{th1} and Corollary~\ref{cor1} can be justified in the remaining stages of the backward dynamic programming by using the same analysis.
\end{proof}

\subsection{Proofs of Theorem~\ref{th3} and Corollary~\ref{cor3}}
\begin{proof}
We prove the result by verifying that the given strategies form an SPE. The dominant pure strategy of the agent and the value function in the final stage are as shown in the proofs of Theorem~\ref{th1} and Corollary~\ref{cor1}. Note that any behavioral adversarial strategy satisfying (\ref{eq14})-(\ref{eq15}) can be $g_{N}^{*}$ since it has no impact on the agent reward. Therefore, Theorem~\ref{th3} and Corollary~\ref{cor3} hold in the final stage.

For stage $N-1$, we first show that it is sufficient to consider a pure agent strategy with an affine form. A general behavioral agent strategy $f_{N-1}$ decides an action $a_{N-1}$ based on the belief $b_{N-1}$ and the observation $\hat{s}_{N-1}$ with the probability measure $f_{N-1}(a_{N-1}|b_{N-1},\hat{s}_{N-1})$. Given a belief $b_{N-1}$, a behavioral adversarial strategy $g_{N-1}$, an observation $\hat{s}_{N-1}$, and an action $a_{N-1}$ from the support set of a behavioral agent strategy $f_{N-1}$, the Q-function is
\begin{align}
Q_{N-1}^{N}&(b_{N-1},g_{N-1},\hat{s}_{N-1},a_{N-1})\nonumber\\
=&\,-\theta_{N-1}E_{b_{N-1}}(S_{N-1}^{2})-\phi_{N-1}a_{N-1}^{2}\nonumber\\
&\,-\tilde{\theta}_{N}E_{g_{N-1}}(\Lambda_{\mu}(b_{N-1},\Pi_{N-1},\Delta_{N-1}^{2},\hat{s}_{N-1},a_{N-1}))^{2}\nonumber\\
&\,-\check{\theta}_{N}E_{g_{N-1}}(\Lambda_{\nu}(b_{N-1},\Pi_{N-1},\Delta_{N-1}^{2})\nonumber)\\
=&\,-\theta_{N-1}(\mu_{N-1}^{2}+\sigma_{N-1}^{2})-(\phi_{N-1}+\tilde{\theta}_{N}\beta_{N-1}^{2})a_{N-1}^{2}\nonumber\\
&\,-2\tilde{\theta}_{N}E_{g_{N-1}}\left(\frac{\Pi_{N-1}\sigma_{N-1}^{2}\hat{s}_{N-1}+\mu_{N-1}\Delta_{N-1}^{2}}{\Pi_{N-1}^{2}\sigma_{N-1}^{2}+\Delta_{N-1}^{2}}\right)\nonumber\\
&\quad\;\alpha_{N-1}\beta_{N-1}a_{N-1}\nonumber\\
&\,-\tilde{\theta}_{N}E_{g_{N-1}}\left(\frac{\Pi_{N-1}\sigma_{N-1}^{2}\hat{s}_{N-1}+\mu_{N-1}\Delta_{N-1}^{2}}{\Pi_{N-1}^{2}\sigma_{N-1}^{2}+\Delta_{N-1}^{2}}\right)^2\alpha_{N-1}^{2}\nonumber\\
&\,-\check{\theta}_{N}E_{g_{N-1}}\left(\frac{\alpha_{N-1}^{2}\sigma_{N-1}^{2}\Delta_{N-1}^{2}}{\Pi_{N-1}^{2}\sigma_{N-1}^{2}+\Delta_{N-1}^{2}}\right)-\check{\theta}_{N}\omega_{N-1}^{2},
\label{a16}
\end{align}
which is a concave quadratic function of $a_{N-1}$. As the best response to maximize the agent reward, the support set of the behavioral agent strategy is a singleton, i.e., it is sufficient to use a pure agent strategy, which has an affine form as
\begin{align}
a_{N-1}=&\,-\frac{\tilde{\theta}_{N}\alpha_{N-1}\beta_{N-1}E_{g_{N-1}}\left(\frac{\Pi_{N-1}\sigma_{N-1}^{2}}{\Pi_{N-1}^{2}\sigma_{N-1}^{2}+\Delta_{N-1}^{2}}\right)}{\phi_{N-1}+\tilde{\theta}_{N}\beta_{N-1}^{2}}\hat{s}_{N-1}\nonumber\\
&\,-\frac{\tilde{\theta}_{N}\alpha_{N-1}\beta_{N-1}E_{g_{N-1}}\left(\frac{\Delta_{N-1}^{2}}{\Pi_{N-1}^{2}\sigma_{N-1}^{2}+\Delta_{N-1}^{2}}\right)}{\phi_{N-1}+\tilde{\theta}_{N}\beta_{N-1}^{2}}\mu_{N-1}.
\label{a17}
\end{align}

Assume that an SPE in behavioral strategies consists of $\kappa_{N-1}^{*}=f_{N-1}^{*}(b_{N-1})=0$ and $\rho_{N-1}^{*}=f_{N-1}^{*}(b_{N-1})=-\frac{\tilde{\theta}_{N}\alpha_{N-1}\beta_{N-1}\mu_{N-1}}{\phi_{N-1}+\tilde{\theta}_{N}\beta_{N-1}^{2}}$. Given $b_{N-1}$, $(\pi_{N-1},\delta_{N-1}^{2})$ in the support set of a behavioral adversarial strategy $g_{N-1}$, $\kappa_{N-1}^{*}$, and $\rho_{N-1}^{*}$, we have the following Q-function:
\begin{align}
Q_{N-1}^{N}&(b_{N-1},\pi_{N-1},\delta_{N-1}^{2},\kappa_{N-1}^{*},\rho_{N-1}^{*})\nonumber\\
=&\,-\left(\theta_{N-1}+\tilde{\theta}_{N}\alpha_{N-1}^{2}-\frac{\tilde{\theta}_{N}^{2}\alpha_{N-1}^{2}\beta_{N-1}^{2}}{\phi_{N-1}+\tilde{\theta}_{N}\beta_{N-1}^{2}}\right)\mu_{N-1}^{2}\nonumber\\
&\,-(\theta_{N-1}+\check{\theta}_{N}\alpha_{N-1}^{2})\sigma_{N-1}^{2}-\check{\theta}_{N}\omega_{N-1}^{2}\nonumber\\
&\,+(\check{\theta}_{N}-\tilde{\theta}_{N})\alpha_{N-1}^{2}\frac{\pi_{N-1}^{2}\sigma_{N-1}^{2}}{\pi_{N-1}^{2}\sigma_{N-1}^{2}+\delta_{N-1}^{2}}\sigma_{N-1}^{2}.
\label{a18}
\end{align}
From Property~\ref{prp4} and the adversarial constraints (\ref{eq14})-(\ref{eq15}), we have
\begin{align}
&\left(\pi_{N-1}\not=0,\delta_{N-1}^{2}=\frac{\pi_{N-1}^{2}\sigma_{N-1}^{2}}{\lambda-1}\right)\nonumber\\
&\;=\arg\min_{\left(\pi_{N-1},\delta_{N-1}^{2}\right)}Q_{N-1}^{N}(b_{N-1},\pi_{N-1},\delta_{N-1}^{2},\kappa_{N-1}^{*},\rho_{N-1}^{*}).
\label{a19}
\end{align}
Therefore, any behavioral adversarial strategy is the best response of $f_{N-1}^{*}$ if its support set consists of two or more elements of $\left(\pi_{N-1}\not=0,\delta_{N-1}^{2}=\frac{\pi_{N-1}^{2}\sigma_{N-1}^{2}}{\lambda-1}\right)$.

Assume that an SPE consists of a behavioral adversarial strategy $g_{N-1}^{*}(\cdot|b_{N-1})$, which is defined on a support set containing two or more elements of $\left(\pi_{N-1}\not=0,\delta_{N-1}^{2}=\frac{\pi_{N-1}^{2}\sigma_{N-1}^{2}}{\lambda-1}\right)$, and satisfies $E_{g_{N-1}^{*}}(\Pi_{N-1})=0$. Given $b_{N-1}$, $g_{N-1}^{*}$, $\kappa_{N-1}$, and $\rho_{N-1}$, we have the following Q-function:
\begin{align}
Q_{N-1}^{N}&(b_{N-1},g_{N-1}^{*},\kappa_{N-1},\rho_{N-1})\nonumber\\
=&\,-(\theta_{N-1}+\tilde{\theta}_{N}\alpha_{N-1}^{2})\mu_{N-1}^{2}-\check{\theta}_{N}\omega_{N-1}^{2}\nonumber\\
&\,-\left(\theta_{N-1}+\check{\theta}_{N}\alpha_{N-1}^{2}-(\check{\theta}_{N}-\tilde{\theta}_{N})\alpha_{N-1}^{2}\frac{\lambda-1}{\lambda}\right)\sigma_{N-1}^{2}\nonumber\\
&\,-(\phi_{N-1}+\tilde{\theta}_{N}\beta_{N-1}^{2})E_{g_{N-1}^{*}}(\Pi_{N-1}^{2})\nonumber\\
&\quad\left(\mu_{N-1}^{2}+\frac{\lambda}{\lambda-1}\sigma_{N-1}^{2}\right)\kappa_{N-1}^{2}\nonumber\\
&\,-(\phi_{N-1}+\tilde{\theta}_{N}\beta_{N-1}^{2})\rho_{N-1}^{2}\nonumber\\
&\,-2\tilde{\theta}_{N}\alpha_{N-1}\beta_{N-1}\mu_{N-1}\rho_{N-1}.
\label{a20}
\end{align}
The best response of $g_{N-1}^{*}$ is
\begin{align}
&\left(\kappa_{N-1}=0,\rho_{N-1}=-\frac{\tilde{\theta}_{N}\alpha_{N-1}\beta_{N-1}\mu_{N-1}}{\phi_{N-1}+\tilde{\theta}_{N}\beta_{N-1}^{2}}\right)\nonumber\\
&\;=\arg\max_{(\kappa_{N-1},\rho_{N-1})}Q_{N-1}^{N}(b_{N-1},g_{N-1}^{*},\kappa_{N-1},\rho_{N-1}).
\label{a21}
\end{align}

It follows from (\ref{a19}) and (\ref{a21}) that a behavioral adversarial strategy $g_{N-1}^{*}(\cdot|b_{N-1})$, which is defined on a support set containing two or more elements of $\left(\pi_{N-1}\not=0,\delta_{N-1}^{2}=\frac{\pi_{N-1}^{2}\sigma_{N-1}^{2}}{\lambda-1}\right)$ and satisfies $E_{g_{N-1}^{*}}(\Pi_{N-1})=0$, and a pure agent strategy $(\kappa_{N-1}^{*},\rho_{N-1}^{*})=f_{N-1}^{*}(b_{N-1})=\left(0,-\frac{\tilde{\theta}_{N}\alpha_{N-1}\beta_{N-1}\mu_{N-1}}{\phi_{N-1}+\tilde{\theta}_{N}\beta_{N-1}^{2}}\right)$ form an SPE in stage $N-1$. Furthermore, the value function in this stage is
\begin{align*}
V_{N-1}^{N}(b_{N-1})=-\tilde{\theta}_{N-1}\mu_{N-1}^{2}-\check{\theta}_{N-1}\sigma_{N-1}^{2}-\check{\theta}_{N}\omega_{N-1}^{2}.
\end{align*}
Thus, Theorem~\ref{th3} and Corollary~\ref{cor3} hold in stage $N-1$.

In the remaining stages of the backward dynamic programming we can always justify that the solution of the SPE in behavioral strategies from Theorem~\ref{th3} and the value function from Corollary~\ref{cor3} hold following the same analysis as used in stage $N-1$.
\end{proof}

\subsection{Proofs of Theorems~\ref{th5} and \ref{th6}}
\begin{proof}
To prove Theorems~\ref{th5} and \ref{th6}, it is sufficient to consider a two-stage problem, i.e., $N=2$.
The solution of the final stage is the same as in the proofs of Theorem~\ref{th3} and Corollary~\ref{cor3}, and therefore is omitted here. Theorems~\ref{th5} and \ref{th6} hold in the final stage.
Furthermore, as shown in the proofs of Theorem~\ref{th3} and Corollary~\ref{cor3}, it is sufficient to consider a pure agent strategy with an affine form in the first stage.

Given $b_{1}$, $(\pi_{1},\delta_{1}^{2})$ in the support set of a behavioral adversarial strategy $g_{1}$, $\kappa_{1}$, and $\rho_{1}$, the Q-function is
\begin{align}
Q_{1}^{2}&(b_{1},\pi_{1},\delta_{1}^{2},\kappa_{1},\rho_{1})\nonumber\\
=&\,-(\theta_{1}+\theta_{2}\alpha_{1}^{2})\mu_{1}^{2}-\theta_{2}\omega_{1}^{2}-(\theta_{1}+\theta_{2}\alpha_{1}^{2})\sigma_{1}^{2}\nonumber\\
&\,-(\phi_{1}+\theta_{2}\beta_{1}^{2})(\pi_{1}\kappa_{1}\mu_{1}+\rho_{1})^{2}-2\theta_{2}\alpha_{1}\beta_{1}\mu_{1}(\pi_{1}\kappa_{1}\mu_{1}+\rho_{1})\nonumber\\
&\,-(\phi_{1}+\theta_{2}\beta_{1}^{2})\kappa_{1}^{2}(\pi_{1}^{2}\sigma_{1}^{2}+\delta_{1}^{2})-2\theta_{2}\alpha_{1}\beta_{1}\pi_{1}\kappa_{1}\sigma_{1}^{2}.
\label{a23}
\end{align}

This Q-function is non-increasing in $\delta_{1}^{2}$ for any given $b_{1}$, $\pi_{1}$, $\kappa_{1}$, and $\rho_{1}$. As the best response to minimize the agent reward, it is sufficient to consider a behavioral adversarial strategy defined on a non-singleton support set of $\left(\pi_{1}\not=0,\delta_{1}^{2}=\frac{\pi_{1}^{2}\sigma_{1}^{2}}{\lambda-1}\right)$.

Given $b_{1}$, $g_{1}$ with a non-singleton support set of $\left(\pi_{1}\not=0,\delta_{1}^{2}=\frac{\pi_{1}^{2}\sigma_{1}^{2}}{\lambda-1}\right)$, $\kappa_{1}$, and $\rho_{1}$, the Q-function is
\begin{align}
Q_{1}^{2}&(b_{1},g_{1},\kappa_{1},\rho_{1})\nonumber\\
=&\,-(\theta_{1}+\theta_{2}\alpha_{1}^{2})\mu_{1}^{2}-\theta_{2}\omega_{1}^{2}-(\theta_{1}+\theta_{2}\alpha_{1}^{2})\sigma_{1}^{2}\nonumber\\
&\,-(\phi_{1}+\theta_{2}\beta_{1}^{2})E_{g_{1}}(\Pi_{1}^{2})\left(\mu_{1}^{2}+\frac{\lambda}{\lambda-1}\sigma_{1}^{2}\right)\kappa_{1}^{2}\nonumber\\
&\,-2(\phi_{1}+\theta_{2}\beta_{1}^{2})E_{g_{1}}(\Pi_{1})\mu_{1}\kappa_{1}\rho_{1}-(\phi_{1}+\theta_{2}\beta_{1}^{2})\rho_{1}^{2}\nonumber\\
&\,-2\theta_{2}\alpha_{1}\beta_{1}E_{g_{1}}(\Pi_{1})(\mu_{1}^{2}+\sigma_{1}^{2})\kappa_{1}-2\theta_{2}\alpha_{1}\beta_{1}\mu_{1}\rho_{1}.
\label{a24}
\end{align}
This is a concave quadratic function of $\rho_{1}$ when $b_{1}$, $g_{1}$, and $\kappa_{1}$ are fixed. As the best response to maximize the agent reward, we can substitute $\rho_{1}$ with
\begin{align}
\rho_{1}=-E_{g_{1}}(\Pi_{1})\mu_{1}\kappa_{1}-\frac{\theta_{2}\alpha_{1}\beta_{1}\mu_{1}}{\phi_{1}+\theta_{2}\beta_{1}^{2}}.
\label{a25}
\end{align}
Then the Q-function (\ref{a24}) reduces to
\begin{align}
Q_{1}^{2}&(b_{1},g_{1},\kappa_{1})\nonumber\\
=&\,-\left(\theta_{1}+\theta_{2}\alpha_{1}^{2}-\frac{\theta_{2}^{2}\alpha_{1}^{2}\beta_{1}^{2}}{\phi_{1}+\theta_{2}\beta_{1}^{2}}\right)\mu_{1}^{2}-\theta_{2}\omega_{1}^{2}-(\theta_{1}+\theta_{2}\alpha_{1}^{2})\sigma_{1}^{2}\nonumber\\
&\,-(\phi_{1}+\theta_{2}\beta_{1}^{2})\nonumber\\
&\quad\left(E_{g_{1}}(\Pi_{1}^{2})\left(\mu_{1}^{2}+\frac{\lambda}{\lambda-1}\sigma_{1}^{2}\right)-E_{g_{1}}^{2}(\Pi_{1})\mu_{1}^{2}\right)\kappa_{1}^{2}\nonumber\\
&\,-2\theta_{2}\alpha_{1}\beta_{1}E_{g_{1}}(\Pi_{1})\sigma_{1}^{2}\kappa_{1}.
\label{a26}
\end{align}
This is also a concave quadratic function of $\kappa_{1}$ when $b_{1}$ and $g_{1}$ are fixed. As the best response to maximize the agent reward, we can substitute $\kappa_{1}$ with
\begin{align}
\kappa_{1}=\frac{-\theta_{2}\alpha_{1}\beta_{1}E_{g_{1}}(\Pi_{1})\sigma_{1}^{2}}{(\phi_{1}+\theta_{2}\beta_{1}^{2})\left(E_{g_{1}}(\Pi_{1}^{2})\left(\mu_{1}^{2}+\frac{\lambda}{\lambda-1}\sigma_{1}^{2}\right)-E_{g_{1}}^{2}(\Pi_{1})\mu_{1}^{2}\right)}.
\end{align}
Since we consider $0\leq\varepsilon'<\varepsilon$ or $\varepsilon'<\varepsilon\leq0$ and the behavioral adversarial strategy has a non-singleton support set, $E_{g_{1}}(\Pi_{1})\not=0$ and $\kappa_{1}\not=0$ in these cases.

We then study the support set of a behavioral adversarial strategy. Given $b_{1}$, $\left(\pi_{1}\not=0,\delta_{1}^{2}=\frac{\pi_{1}^{2}\sigma_{1}^{2}}{\lambda-1}\right)$ in the support set of a behavioral adversarial strategy $g_{1}$, $\kappa_{1}\not=0$, and $\rho_{1}$, the Q-function is
\begin{align}
Q_{1}^{2}&(b_{1},\pi_{1},\kappa_{1},\rho_{1})\nonumber\\
=&\,-(\theta_{1}+\theta_{2}\alpha_{1}^{2})\mu_{1}^{2}-\theta_{2}\omega_{1}^{2}-(\theta_{1}+\theta_{2}\alpha_{1}^{2})\sigma_{1}^{2}\nonumber\\
&\,-(\phi_{1}+\theta_{2}\beta_{1}^{2})(\pi_{1}\kappa_{1}\mu_{1}+\rho_{1})^{2}-2\theta_{2}\alpha_{1}\beta_{1}\mu_{1}(\pi_{1}\kappa_{1}\mu_{1}+\rho_{1})\nonumber\\
&\,-(\phi_{1}+\theta_{2}\beta_{1}^{2})\kappa_{1}^{2}\frac{\lambda}{\lambda-1}\pi_{1}^{2}\sigma_{1}^{2}-2\theta_{2}\alpha_{1}\beta_{1}\pi_{1}\kappa_{1}\sigma_{1}^{2}.
\end{align}
This is a concave quadratic function of $\pi_{1}$ given $b_{1}$, $\kappa_{1}\not=0$, and $\rho_{1}$. As the best response to minimize the agent reward, it is sufficient to consider a behavioral adversarial strategy $g_{1}$ with the following support set: $\left\{\left(\pi_{1}=\varepsilon',\delta_{1}^{2}=\frac{\varepsilon'^{2}\sigma_{1}^{2}}{\lambda-1}\right),\left(\pi_{1}=\varepsilon,\delta_{1}^{2}=\frac{\varepsilon^{2}\sigma_{1}^{2}}{\lambda-1}\right)\right\}$.

When $0=\varepsilon'<\varepsilon$ or $\varepsilon'<\varepsilon=0$, an SPE in behavioral strategies does not exist since $0=\varepsilon'$ or $\varepsilon=0$ will lead to a singleton support set of the behavioral adversarial strategy; and meanwhile a pure strategy SPE does not exist since $\varepsilon'\not=\varepsilon$. This proves Theorem~\ref{th5}.

When $0<\varepsilon'<\varepsilon$ or $\varepsilon'<\varepsilon<0$, we assume that an SPE in behavioral strategies consists of
\begin{gather*}
g_{1}^{*}\left(\pi_{1}=\varepsilon',\left.\delta_{1}^{2}=\frac{\varepsilon'^{2}\sigma_{1}^{2}}{\lambda-1}\right|b_{1}\right)=p^{*};\\
g_{1}^{*}\left(\pi_{1}=\varepsilon,\left.\delta_{1}^{2}=\frac{\varepsilon^{2}\sigma_{1}^{2}}{\lambda-1}\right|b_{1}\right)=1-p^{*};\\
\kappa_{1}^{*}=f_{1}^{*}(b_{1})\qquad\qquad\qquad\qquad\qquad\qquad\qquad\qquad\qquad\qquad\qquad\qquad\\
\;\;\;=\frac{-\theta_{2}\alpha_{1}\beta_{1}E_{g_{1}^{*}}(\Pi_{1})\sigma_{1}^{2}}{(\phi_{1}+\theta_{2}\beta_{1}^{2})\left(E_{g_{1}^{*}}(\Pi_{1}^{2})\left(\mu_{1}^{2}+\frac{\lambda}{\lambda-1}\sigma_{1}^{2}\right)-E_{g_{1}^{*}}^{2}(\Pi_{1})\mu_{1}^{2}\right)};\\
\rho_{1}^{*}=f_{1}^{*}(b_{1})=-E_{g_{1}^{*}}(\Pi_{1})\mu_{1}\kappa_{1}^{*}-\frac{\theta_{2}\alpha_{1}\beta_{1}\mu_{1}}{\phi_{1}+\theta_{2}\beta_{1}^{2}}.
\end{gather*}
Since the assumed $f_{1}^{*}$ is the best response of the assumed $g_{1}^{*}$, we only need to testify if $0<p*<1$ exists such that $g_{1}^{*}$ is the best response of $f_{1}^{*}$, i.e., both $\pi_{1}=\varepsilon'$ and $\pi_{1}=\varepsilon$ are minimizers of the Q-function $Q_{1}^{2}(b_{1},\pi_{1},\kappa_{1}^{*},\rho_{1}^{*})$. There is a unique solution
\begin{align}
p^{*}=\frac{\varepsilon}{\varepsilon'+\varepsilon}.
\end{align}
Therefore, there is a unique SPE in behavioral strategies for the two-stage ALQG game. The strategies and the value function of the SPE in Theorem~\ref{th6} can then be obtained easily.
\end{proof}

\end{document}